\documentclass[lettersize,journal]{IEEEtran}
\usepackage{amsmath,amsfonts,amssymb,amsthm}
\usepackage{mathtools}

\usepackage{graphicx}
\usepackage[caption=false,font=normalsize]{subfig} 
\usepackage{dblfloatfix}
\usepackage{array}
\usepackage{multirow}

\usepackage{algorithm}
\usepackage{algpseudocode}

\usepackage{cite}
\usepackage{url}

\usepackage{xcolor} 
\usepackage{comment}
\usepackage{makecell}

\DeclareMathOperator*{\argmax}{argmax}

\begin{document}
\bstctlcite{IEEEexample:BSTcontrol}
\newtheorem{thm}{Theorem}
\newtheorem{remark}{Remark}
\newtheorem{lemma}{Lemma}
\newtheorem{prop}{Proposition}
\newtheorem{defn}{Definition}
\newtheorem{condi}{Condition}
\newtheorem{assump}{Assumption}

\title{Markov Potential Game and Multi-Agent Reinforcement Learning for Autonomous Driving}

\author{Huiwen Yan\IEEEauthorrefmark{1} and Mushuang Liu\IEEEauthorrefmark{1}~\IEEEmembership{, Member, IEEE}

\thanks{\IEEEauthorrefmark{1} H. Yan and M. Liu are with the Department of Mechanical Engineering at Virginia Tech.
}

\thanks{This work was supported by DARPA Young Faculty Award with the grant number D24AP00321.}
}

\markboth{}%
{Liu \MakeLowercase{\textit{et al.}}: Markov Potential Game based Multi-Agent Reinforcement Learning for Autonomous Driving}

\IEEEpubid{0000--0000/00\$00.00~\copyright~2025 IEEE}

\maketitle

\begin{abstract}
Autonomous driving (AD) requires safe and reliable decision-making among interacting agents, e.g., vehicles, bicycles, and pedestrians. Multi-agent reinforcement learning (MARL) modeled by Markov games (MGs) provides a suitable framework to characterize such agents' interactions during decision-making. Nash equilibria (NEs) are often the desired solution in an MG. However, it is typically challenging to compute an NE in general-sum games, unless the game is a Markov potential game (MPG), which ensures the NE attainability under a few learning algorithms such as gradient play. However, it has been an open question how to construct an MPG and whether these construction rules are suitable for AD applications. In this paper, we provide sufficient conditions under which an MG is an MPG and show that these conditions can accommodate general driving objectives for autonomous vehicles (AVs) using highway forced merge scenarios as illustrative examples. A parameter-sharing neural network (NN) structure is designed to enable decentralized policy execution. The trained driving policy from MPGs is evaluated in both simulated and naturalistic traffic datasets. Comparative studies with single-agent RL and with human drivers whose behaviors are recorded in the traffic datasets are reported, respectively. 
\end{abstract}


\section{Introduction}

To realize intelligent decision-making in autonomous driving (AD), reinforcement learning (RL) has been recently explored \cite{RLinAD1, RLinAD2}. The RL problem is often modeled as a Markov decision process (MDP), where an autonomous vehicle (AV) acts as the agent and interacts with the environment to learn a driving policy. However, real traffic scenarios involve multiple decision-making vehicles whose actions mutually influence one another. Treating all other vehicles as part of the environment and optimizing the ego vehicle’s policy via single-agent RL neglects the fact that surrounding vehicles are also active decision-makers with possibly non-stationary policies \cite{PlanningMultivehicle, I2QBetterIQ}, thereby failing to account for strategic interactions among traffic participants. This could lead to risky or overly conservative ego vehicle behavior \cite{MARLforADSurvey, MARLChallenge}.

To capture strategic interactions among multiple traffic agents, multi-agent reinforcement learning (MARL) has been recently adopted in AD \cite{GamePlanInAD, GameSocialInAD, GameforVerificationInAD}. In particular, cooperative MARL, under the assumption of a shared global objective among agents, has been extensively studied \cite{DecPOMDP}. Many methods demonstrate strong empirical performance in coordinated driving scenarios \cite{CollabAware, MADRLIntersection}. Nevertheless, such formulations may raise other concerns. First, the assumption of a shared global objective may not necessarily hold in real-world AD scenarios, where vehicles are often independently operated while optimizing self-centered heterogeneous objectives. Second, theoretical guarantees, e.g., algorithmic convergence, would be especially crucial for safety-critical applications such as AD; however, many cooperative MARL methods under the centralized training and decentralized execution (CTDE) paradigm rely primarily on empirical evaluations \cite{SafeLearning}. Unlike cooperative settings, a Markov game (MG) formulation models all traffic participants as active decision-makers with coupled self-interests, e.g., in the scenarios of merging, platooning, intersection crossing, and mixed-traffic interactions \cite{2PMerge, MicroLevel, socialLearning, GenerativeWorldModel, HMmixedTraffic}.

A Nash equilibrium (NE) is often a desired solution in a MG, which represents a stationary status where no agent can improve its total reward by unilaterally deviating from its policy \cite{fudenberg1991game}. However, most of the current MARL algorithms establish the NE attainability under very restrictive game structure assumptions, e.g., zero-sum or identical-interest MGs for Nash Q-learning  \cite{TwoPlayerZeroSum, symmetricTeams}. These special game structures hardly hold in AD applications where traffic agents' objective functions are hardly zero-sum or identical. For the MARL algorithms that target general-sum MGs, such as gradient play and no-regret learning, the attainability of an NE (including the convergence of the algorithm and the optimality of the converged point) can be guaranteed if the MG is a Markov potential game (MPG) \cite{gradientPlay,noRegretLearning}.

\IEEEpubidadjcol 
MPGs are a special class of MGs with appealing properties. In general, a potential game means that there exists a function, called a potential function, such that any unilateral improvement by an agent's policy deviation equals the improvement in the potential function \cite{la2016potential}. If a game is a potential game, then a pure-strategy NE always exists, and most NE-seeking algorithms, such as best- and better- response dynamics, are guaranteed to converge \cite{MONDERER1996124}. Due to these appealing properties, potential games have been recently adopted in various AD applications, including intersection-crossing, highway-merging, and incoming traffic scenarios \cite{potentialGamesMLiu, PCPGMLiu, gmiPotentialGame, gameProjectionMLiu}. Of particular interest to our line of research, the reward shaping approach in \cite{potentialGamesMLiu} states that if agents' reward functions satisfy certain conditions, then the resulting game is a potential game. These conditions can accommodate many traffic scenarios and driving considerations, such as tracking a desired speed while avoiding collisions. However, the existing studies only focus on static games at each time instant. Whether they can be generalized to MGs has been an open question. As pointed out by \cite{gradientPlay,globalConvergence}, the MPG construction could be much more challenging than static potential game construction.

In this paper, we provide sufficient conditions on the reward design and on the Markov decision process such that the resulting MG is an MPG. We also demonstrate that these conditions can accommodate common AD scenarios using highway forced merge as illustrative examples. Specifically, the contributions of our paper are fourfold:
\begin{enumerate}
    \item We develop sufficient conditions on the reward design and on the Markov transition properties such that an MG is an MPG.
    \item We show that the developed conditions can accommodate general driving objectives (e.g., tracking a desired speed, collision avoidance, or maximizing ride comfort) in AD applications. 
    \item We design the policy neural network using a parameter-sharing strategy  to address the scalability issue and to improve the training efficiency. 
    \item We  perform comprehensive numerical studies on both simulated and naturalistic traffic datasets and show that the MPG-based MARL approach outperforms single-agent RL and human drivers under certain metrics, including safety, energy efficiency, and ride comfort.
\end{enumerate}

The paper is organized as follows: Section II formulates the AD problem as an MG problem and reviews MARL preliminaries. Section III defines MPGs and their properties, and presents the MPG construction method. Section IV introduces the gradient-based MARL algorithm to solve the MPGs and the parameter-sharing policy network architecture. Section V reports simulation results in highway forced merge scenarios, including statistical and comparative studies. Section VI concludes the paper.

\section{Problem formulation}
We formulate the AV decision-making problems as Markov games in Section \ref{subsec:MG} and consider MARL in Section \ref{subsec:MARL} to solve the formulated Markov games.
\subsection{Markov Games}\label{subsec:MG}
Let's consider a traffic scenario consisting of a set of traffic agents $\mathcal{N}=\{1,2,\cdots, N\}$. We model the dynamics of each vehicle $i\in\mathcal{N}$ using the kinematic bicycle model\cite{Rajamani2012}:
\begin{equation}\label{eq:vehicleDynamics}
\begin{split}
    x_{i,t+1} &= x_{i,t} + v_{i,t} \cos\bigl(\phi_{i,t} + \beta_{i,t}\bigr)\,\Delta t,\\
    y_{i,t+1} &= y_{i,t} + v_{i,t} \sin\bigl(\phi_{i,t} + \beta_{i,t}\bigr)\,\Delta t,\\
    v_{i,t+1} &= v_{i,t} + u_{i,t}\,\Delta t,\\
    \phi_{i,t+1} &= \phi_{i,t} + \frac{v_{i,t}}{l_r} \sin\bigl(\beta_{i,t}\bigr)\,\Delta t,\\
    \beta_{i,t} &= \arctan\Biggl(\frac{l_r}{l_r + l_f}\tan\bigl(\delta_{i,f,t}\bigr)\Biggr).
\end{split}
\end{equation}
Here, $x_i$ and $y_i$ denote the longitudinal and lateral positions of vehicle $i$'s center of mass (CoM) in a $x\text{-}y$ inertial frame. The variables $v_i$ and $u_i$ represent the speed and acceleration of vehicle $i$'s CoM, respectively. The heading angle relative to the inertial frame is denoted by $\phi_i$, while $\beta_i$ denotes the slip angle, defined as the deviation between the speed and the longitudinal axis of the vehicle. The parameters $l_f$ and $l_r$ denote the distances from the CoM of the vehicle to its front and rear axles, and $\delta_{i,f}$ is the front-wheel steering angle. The dynamics are discretized by a fixed sampling time $\Delta t$. At time instant $t$, the state of vehicle $i$ is given by $s_{i,t}
=
[x_{i,t},\, y_{i,t},\, v_{i,t},\, \phi_{i,t}]^{\top}
\in \mathcal{S}_i$, where $\mathcal{S}_i$ denotes the state space. The state information is assumed to be accessible through onboard sensors. The action is $a_{i,t}
=
[u_{i,t},\, \delta_{i,f,t}]^{\top}
\in \mathcal{A}_i$, consisting of acceleration and front-wheel steering angle. The global state is $s_{t}=(s_{1,t}, \cdots, s_{N,t})$ and a joint action is $a_{t}=(a_{1,t}, \cdots, a_{,t})$. Specifically, $s_0\sim\rho$ denotes the initial global state, where $\rho$ is the distribution of the initial global state. The global state space is given by
$\mathcal{S} = \mathcal{S}_1 \times \cdots \times \mathcal{S}_N$,
and the joint action space is defined as
$\mathcal{A} = \mathcal{A}_1 \times \cdots \times \mathcal{A}_N$.

Each vehicle in this traffic system has its own driving objective, which may consist of collision avoidance, speed regulation, and passenger comfort. Given a state-action pair at time $t$, we denote the reward of vehicle $i$ as $r_i(s_{i,t},s_{-i,t},a_{i,t},a_{-i,t})$, where $r_i:\mathcal{S} \times \mathcal{A} \rightarrow \mathbb{R}$ is the reward function and the subscript $-i$ represents the set $\mathcal{N}\backslash\{i\}$. Each vehicle solves for its own optimal action sequence, i.e.,
$\mathbf{a}_{i}^{*}
=
\{a_{i,t}^{*}\}_{t=0}^{\infty}
\in \mathcal{U}_i$, where $\mathcal{U}_i$ is the space of the action sequence. If the initial state $s_0$ is fixed, the optimal action sequence can be given by
\begin{equation}\label{eq:optimizer}
\mathbf{a}_{i}^{*}
=
\argmax_{\mathbf{a}_{i}\in\mathcal{U}_i}
\sum_{t=0}^{\infty}\gamma^t
r_i\left(
s_{i,t},\, a_{i,t},\, s_{-i,t},\, a_{-i,t}
\right),
\end{equation}
where $\gamma\in[0,1)$ is a discount factor.

As shown in \eqref{eq:optimizer}, the optimal action sequence $\mathbf{a}_{i}^{*}$ of vehicle $i$ depends on the actions of other agents through the coupled reward $r_i(\cdot)$, characterizing agents' interactions. Such a coupled optimization problem leads to a MG problem since the state transition \eqref{eq:vehicleDynamics} satisfies the Markov property. 

In general, a Markov game is defined by the tuple
$\mathcal{M} = (\mathcal{N}, \mathcal{S}, \mathcal{A}, P, r, \gamma, \rho)$.
For each agent $i$ at time $t$, the Markov state and action are identified with the
vehicle state and action, respectively. Given a global state $s_{t}$ and a joint action $a_{t}$, the probability of the traffic system evolving to the next state $s_{t+1}$ is given by $P(s_{t+1}|s_{t}, a_{t})$, determined by the state transition model \eqref{eq:vehicleDynamics}. Collectively, the system is characterized by the reward vector $r=(r_1,\cdots,r_N)$. The discount factor $\gamma$ introduces the trade-off between immediate rewards and long-term rewards.

Agents make decisions according to a policy function $\pi_i:\mathcal{S}\rightarrow\Delta(\mathcal{A}_i)$, which maps states to probability distributions over the action space $\mathcal{A}_i$. Agents act independently based on their individual policies. That is, at each time instant $t$, the distribution factorizes across the agents:
\begin{equation}\label{eq:decentralization}
    \mathrm{Pr}(a_t|s_t)=\pi(a_t|s_t)=\prod_{i=1}^{N}\pi_i(a_{i,t}|s_t).
\end{equation}

Let's then specify each agent's policy with parameters $\theta_i$:
\begin{equation}\label{eq:direct}
    \pi_{i,\theta_i}(a_{i,t}|s_t)=\theta_{i,(s,a_i)},\quad i=1,2,\cdots,N.
\end{equation}

For notational simplicity, and with a slight abuse of notation, we shall use $\theta_i$ and $\theta=(\theta_1, \cdots, \theta_N)$ to denote the parameterized policies $\pi_{i,\theta_i}$ and $\pi_{\theta}$, whenever the context is clear. We require $\theta_i\in\Delta(\mathcal{A}_i)^{|\mathcal{S}|}$, where $|\mathcal{S}|$ denotes the cardinality of the state space. We further denote such sets of the parameters as $\mathcal{X}_i=\Delta(\mathcal{A}_i)^{|\mathcal{S}|}$ for each agent $i$, and the Cartesian product $\mathcal{X}=\mathcal{X}_1\times\cdots\times\mathcal{X}_N$ for the system.

A trajectory is defined as $\tau\coloneqq(s_t,a_t)_{t=0}^{\infty}$, where $a_t\sim\pi_{\theta}(\cdot|s_t)$ and $s_{t}\sim P(\cdot|s_{t-1},a_{t-1})$. Given that the system evolves under policy $\pi_{\theta}$, then the value function of agent $i$ denoted as $V_i^{\theta}:\mathcal{S}\rightarrow\mathbb{R}$, is defined as the expected discounted reward from the initial state $s$:
\begin{equation}\label{eq:valueFunction}
    V_i^{\theta}(s) \coloneqq \mathbb{E}\Big[\sum_{t=0}^{\infty}\gamma ^tr_i(s_t,a_t)\Bigm|\pi_{\theta},s_0=s\Big],
\end{equation}
where the expectation is taken with respect to the trajectory distribution induced by $\pi_\theta$ and the transition kernel $P$.

For each agent $i$, we define its total reward $J_i:\mathcal{X}\rightarrow\mathbb{R}$ as:
\begin{equation}\label{eq:objectiveFunction}
    J_i(\theta)=J_i(\theta_i,\theta_{-i})\coloneqq\mathbb{E}_{s_0\sim\rho}V_i^{\theta}(s_0),
\end{equation}
where the expectation is taken with respect to the randomness of the initial state distribution $\rho$.

From a game-theoretic perspective, an NE is a particularly desirable outcome, for it represents a stable status in multi-agent interactions, where each agent's policy is a best response to the policies of others.

\begin{defn}\label{df:nashEquilibrium}(Nash equilibrium \cite{fudenberg1991game})
    A policy $\theta^*=(\theta_1^*,\cdots,\theta_N^*)$ is called a Nash equilibrium if
    \begin{equation}\label{ineq:nashEquilibrium}
    J_i(\theta_i^*,\theta_{-i}^*)\geq{J_i(\theta_i',\theta_{-i}^*)},\quad\forall \theta_i' \in \mathcal{X}_i,\quad i \in \mathcal{N}.
    \end{equation}
\end{defn}

Furthermore, if the NE $\theta^{*}$ corresponds to a deterministic policy, it is also called a pure-strategy NE; otherwise, it is a mixed-strategy NE. Following \cite{agarwal2019reinforcement}, we define the discounted state visitation measure $d_{\theta}$ over the state space as:
\begin{equation}\label{eq:visitationMeasure}
    d_{\theta}(s)\coloneqq\mathbb{E}_{s_0\sim\rho}(1-\gamma)\sum_{t=0}^{\infty}\gamma^{t}\mathrm{Pr}_{\theta}(s_t=s|s_0),
\end{equation}\\
where $\mathrm{Pr}_{\theta}(s_t=s|s_0)$ denotes the probability of state $s$ being visited at time
$t$ given the initial state $s_0$ under policy $\pi_{\theta}$. Then the discounted state-action visitation measure can be defined as $\mu_{\theta}(s,a)\coloneqq d_{\theta}(s)\pi_{\theta}(a|s)$.

We make the following assumption throughout this work.
\begin{assump}\label{as:positiveProbability}
The Markov game $\mathcal{M}$ satisfies: $d_{\theta}(s)>0, \forall s \in \mathcal{S}, \forall \theta \in \mathcal{X}$.
\end{assump}
This assumption ensures that every state in the state space is visited with a positive probability, a standard assumption in RL convergence proofs \cite{globalConvergence}. Note that this assumption can be readily satisfied if the initial distribution is positive over the state space, i.e., $\rho(s)>0, \forall s\in\mathcal{S}$.

\subsection{Multi-agent reinforcement learning}\label{subsec:MARL}
Given a direct distributed parametrization \eqref{eq:decentralization} and \eqref{eq:direct}, we employ the gradient play for each agent \cite{gradientPlay}:
\begin{equation}\label{eq:gradientPlay}
    \theta_i^{(k+1)}=\text{Proj}_{\mathcal{X}_i}(\theta_i^{(k)}+\eta\nabla_{\theta_i}J_i(\theta^{(k)})),\quad\eta>0,
\end{equation}
where $k$ indexes the learning iteration and $\eta>0$ is the step size.

\begin{defn}\label{df:firstOrderStationaryPolicy}
(First-order stationary policy\cite{gradientPlay}) A policy $\theta^*=(\theta_1^*, \cdots, \theta_N^{*})$ is called a first-order stationary policy if $(\theta_i'-\theta_i^{*})^{\top}\nabla_{\theta_i}J_i(\theta^{*})\leq0,\enspace\forall\theta_i'\in\mathcal{X}_i,\enspace i\in \mathcal{N}$.
\end{defn}

\begin{lemma}\label{le:gradientDomination}
    (Gradient domination\cite{gradientPlay}) For the direct distributed parameterization \eqref{eq:direct}, the following inequality holds for any $\theta=(\theta_1,\cdots,\theta_N)\in \mathcal{X}$ and any $\theta_i'\in\mathcal{X}_i,i\in\mathcal{N}$:
    \begin{equation}\label{eq:gradientDomination}
        J_i(\theta_i',\theta_{-i})-J_i(\theta_i,\theta_{-i})\leq\left\|\frac{d_{\theta'}}{d_\theta}\right\|_{\infty}\max_{\overline{\theta}_i\in\mathcal{X}_i}(\overline{\theta}_i-\theta_i)^{\top}\nabla_{\theta_i}J_i(\theta),
    \end{equation}
    where$\|\frac{d_{\theta'}}{d_\theta}\|_{\infty}\coloneqq\max_s\frac{d_{\theta'}(s)}{d_{\theta}(s)}$, and $\theta'=(\theta_i',\theta_{-i})$.
\end{lemma}

Compared with single-agent cases, the condition \eqref{eq:gradientDomination} becomes weaker because the inequality requires $\theta_{-i}$ to remain fixed, yet it still holds. Even though gradient domination no longer guarantees equivalence between first-order stationarity and global optimality, it still leads to the following equivalence between NEs and the first-order stationary policy.

\begin{thm}\label{thm:equivalence}
    (Theorem 1\cite{gradientPlay}) Under Assumption \ref{as:positiveProbability}, first-order stationary policies and NEs are equivalent.
\end{thm}

While MGs offer a reasonable framework for modeling AD decision-making involving multiple interactive decision-makers, fundamental challenges exist:
\begin{enumerate}
    \item In a general MG, a pure-strategy NE may not always exist \cite{gradientPlay};
    \item The solution-seeking algorithms, e.g., \eqref{eq:gradientPlay}, may not always converge \cite{gradientPlay,noRegretLearning}.
\end{enumerate}

To address the above challenges and to ensure the solvability of the problem, including both solution existence and algorithm convergence, we investigate a special class of MGs, called Markov potential games. 
\section{Markov potential game}
 We first define MPGs and introduce their  properties in Section \ref{subsec:whatAndWhy}, and present sufficient conditions  for the MPG construction in Section \ref{subsec:how}.

\subsection{Definition and Properties of MPGs}\label{subsec:whatAndWhy}
\begin{defn}\label{df:markovPotentialGame}
    (Markov potential game \cite{gradientPlay}) An MG $\mathcal{M}$ becomes an MPG if there exists a potential function $\phi:\mathcal{S}\times\mathcal{A}\rightarrow\mathbb{R}$ such that, $\forall i\in\mathcal{N}, \forall (\theta_i',\theta_{-i}), (\theta_i,\theta_{-i})\in\mathcal{X}, \forall s \in\mathcal{S}$, the following condition holds:
    \begin{equation}\label{eq:markovPotentialGame}
\begin{split}
&\mathbb{E}\Biggl[\sum_{t=0}^{\infty}\gamma^{t}r_i(s_t,a_t)\,\Biggm|\,\pi_{(\theta_i',\theta_{-i})},\,s_0=s\Biggr] \\
&\quad - \mathbb{E}\Biggl[\sum_{t=0}^{\infty}\gamma^{t}r_i(s_t,a_t)\,\Biggm|\,\pi_{(\theta_i,\theta_{-i})},\,s_0=s\Biggr] \\
&= \mathbb{E}\Biggl[\sum_{t=0}^{\infty}\gamma^{t}\phi(s_t,a_t)\,\Biggm|\,\pi_{(\theta_i',\theta_{-i})},\,s_0=s\Biggr] \\
&\quad - \mathbb{E}\Biggl[\sum_{t=0}^{\infty}\gamma^{t}\phi(s_t,a_t)\,\Biggm|\,\pi_{(\theta_i,\theta_{-i})},\,s_0=s\Biggr].
\end{split}
\end{equation}
\end{defn}

Next, we define the total potential function of the MPG as:
\begin{equation}\label{eq:totalPotentialFunction}
\Phi(\theta)\coloneqq\mathbb{E}_{s_0\sim\rho}\left[\sum_{t=0}^{\infty}\gamma^t\phi(s_t,a_t)\Biggm|\pi_{(\theta_i,\theta_{-i})},s_0\right].
\end{equation}

 Theorem \ref{pp:aloGlobalMaximum} guarantees the existence of at least one pure NE in an MPG.
\begin{thm}\label{pp:aloGlobalMaximum}
    (Proposition 1\cite{gradientPlay}) In an MPG, there is at least one global maximum $\theta^{*}$ of the total potential function $\Phi$, i.e., $\theta^{*}\in\argmax_{\theta\in\mathcal{X}}\Phi(\theta)$, that is a pure NE.
\end{thm}

To consider the convergence of the solution-seeking algorithm, let us consider the gradient play \eqref{eq:gradientPlay}.
According to Definition \ref{df:markovPotentialGame}, we know that  $\forall i\in\mathcal{N}, \forall (\theta_i',\theta_{-i}), (\theta_i,\theta_{-i})\in\mathcal{X}$:
\begin{equation}\label{eq:rewrittenMPGDefinition}
    J_i(\theta_i',\theta_{-i})-J_i(\theta_i,\theta_{-i})=\Phi(\theta_i',\theta_{-i})-\Phi(\theta_i,\theta_{-i}).
\end{equation}

With appropriate smoothness assumptions, \eqref{eq:rewrittenMPGDefinition} transforms to
\begin{equation}\label{eq:gradientPlayEquivalence}
    \nabla_{\theta_i}J_i(\theta)=\nabla_{\theta_i}\Phi(\theta).
\end{equation}

Therefore, the gradient play \eqref{eq:gradientPlay} in an MPG is equivalent to
\begin{equation}\label{eq:gradientPlayTotal}
    \theta_i^{(k+1)}=\text{Proj}_{\mathcal{X}_i}(\theta_i^{(k)}+\eta\nabla_{\theta_i}\Phi(\theta^{(k)})),\enspace\eta>0,
\end{equation}

Theorem \ref{thm:globalConvergence} ensures the convergence of \eqref{eq:gradientPlay} and  \eqref{eq:gradientPlayTotal}  to an NE if the MG is an MPG.
\begin{thm}\label{thm:globalConvergence}
    (\cite[Theorem 4.2]{globalConvergence}) Given an MPG, for any initial state, the gradient play in Eq. \eqref{eq:gradientPlay} or \eqref{eq:gradientPlayTotal} converges to an NE as $k\rightarrow\infty$.
\end{thm}

With the above appealing properties, let us see how to construct MPGs.

\subsection{Construction of MPGs}\label{subsec:how}
Theorem \ref{thm:selfPotentialFunction} states that if agent $i$'s state transition and reward are both  determined by solely its own policy, then the resulting MG is an MPG.
\begin{thm}\label{thm:selfPotentialFunction}
    Consider an MG where agent $i$'s reward function is of the form,
    \begin{equation}\label{eq:selfReward}
        r_i(s_t,a_t)=r_i^{\mathrm{self}}(s_{i,t},a_{i,t}),
    \end{equation}
    where $r_i^{\mathrm{self}}(s_{i,t},a_{i,t})$ is solely dependent on agent $i$'s policy $\theta_i$. Suppose $P(s_i'|s_i,a_i,a_{-i}')=P(s_i'|s_i,a_i,a_{-i})$, $\forall a_{-i},a_{-i}'\in\mathcal{A}_{-i},\forall a_i\in\mathcal{A}_i,\forall s_i\in\mathcal{S}_i$ and $\forall i\in\mathcal{N}$. Then the formulated game is an MPG with a potential function \begin{equation}\label{eq:selfPotenitalFunction}
        \phi^{\mathrm{self}}(s_t,a_t)=\sum_{i\in\mathcal{N}}r_i^{\mathrm{self}}(s_{i,t},a_{i,t}).
    \end{equation}
\end{thm}

\begin{proof}
    From \eqref{eq:selfReward}, the total reward of agent $i$ is given by:
\begin{equation}\label{preq:selfTotalReward}
\begin{split}
J_i(\theta)
&= \mathbb{E}_{s_0\sim\rho}\left[\sum_{t=0}^{\infty}\gamma^{t}
\, r_i^{\mathrm{self}}(s_{i,t},a_{i,t}) \,\middle|\, \pi_{\theta},\, s_0\right] \\
&= \frac{1}{1-\gamma}\,
\mathbb{E}_{(s,a)\sim \mu_{\theta}(s,a)}
\left[r_i^{\mathrm{self}}(s_i,a_i)\right].
\end{split}
\end{equation}
    
A unilateral deviation in agent $i$'s policy leads to the following difference in total reward:
    \begin{equation}\label{preq:leftSelfRewardFunction}
\begin{split}
&J_i(\theta_i',\theta_{-i}) - J_i(\theta_i,\theta_{-i}) \\
&= \frac{1}{1-\gamma}\Bigl(
\mathbb{E}_{(s,a)\sim \mu_{\theta_i',\theta_{-i}}}
\left[r_i^{\mathrm{self}}(s_i,a_i)\right] \\
&\qquad\quad
- \mathbb{E}_{(s,a)\sim \mu_{\theta_i,\theta_{-i}}}
\left[r_i^{\mathrm{self}}(s_i,a_i)\right]
\Bigr).
\end{split}
\end{equation}

    From \eqref{eq:selfPotenitalFunction}, the total potential function is given by:
    \begin{equation}\label{preq:selfTotalPotentialFunction}
    \begin{split}
    \Phi(\theta)&=\mathbb{E}_{s_0\sim\rho}\Biggl[\sum_{t=0}^{\infty}\gamma^{t}\sum_{i\in\mathcal{N}}r_i^{\mathrm{self}}(s_{i,t},a_{i,t})\Bigg|\pi_{\theta},s_0\Biggr]\\
    &=\frac{1}{1-\gamma}\,
\mathbb{E}_{(s,a)\sim \mu_{\theta}(s,a)}\Bigg[\sum_{i\in\mathcal{N}}r_i^{\mathrm{self}}(s_i,a_i)\Bigg].
    \end{split}
    \end{equation}
    
    The same deviation shall yield the following difference in the total potential function:
    \begin{equation}\label{preq:rightSelfPotentialFunction}
    \begin{split}
    &\Phi(\theta_i',\theta_{-i})-\Phi(\theta_i,\theta_{-i})\\
    &=\frac{1}{1-\gamma}\Bigg(
\mathbb{E}_{(s,a)\sim \mu_{\theta_i',\theta_{-i}}(s,a)}\Bigg[\sum_{i\in\mathcal{N}}r_i^{\mathrm{self}}(s_i,a_i)\Bigg]\\
&\quad -\mathbb{E}_{(s,a)\sim \mu_{\theta_i,\theta_{-i}}(s,a)}\Bigg[\sum_{i\in\mathcal{N}}r_i^{\mathrm{self}}(s_i,a_i)\Bigg]\Bigg).
    \end{split}
    \end{equation}

    For \eqref{preq:leftSelfRewardFunction} and \eqref{preq:rightSelfPotentialFunction} to align, we need to prove:
    \begin{equation}\label{preq:expectationLevelSelf}
    \begin{split}
        &\mathbb{E}_{(s,a)\sim \mu_{\theta_i',\theta_{-i}}(s,a)}\Bigg[\sum_{j\in\mathcal{N},j\neq i}r_j^{\mathrm{self}}(s_j,a_j)\Bigg]\\
&=\mathbb{E}_{(s,a)\sim \mu_{\theta_i,\theta_{-i}}(s,a)}\Bigg[\sum_{j\in\mathcal{N},j\neq i}r_j^{\mathrm{self}}(s_j,a_j)\Bigg].
    \end{split}
    \end{equation}

    By the law of total expectation, we rewrite \eqref{preq:expectationLevelSelf} as follows:
    \begin{equation}\label{preq:expectationRewriteSelf}
        \begin{split}
            &\mathbb{E}_{(s,a)\sim \mu_{\theta_i',\theta_{-i}}(s,a)}\Bigg[\mathbb{E}\Bigg[\sum_{j\in\mathcal{N},j\neq i}r_j^{\mathrm{self}}(s_j,a_j)\Bigg|s_{-j},a_{-j}\Bigg]\Bigg]\\
&=\mathbb{E}_{(s,a)\sim \mu_{\theta_i,\theta_{-i}}(s,a)}\Bigg[\mathbb{E}\Bigg[\sum_{j\in\mathcal{N},j\neq i}r_j^{\mathrm{self}}(s_j,a_j)\Bigg|s_{-j},a_{-j}\Bigg]\Bigg].
        \end{split}
    \end{equation}
    
Under the assumption of $P(s_j'|s_j,a_j,a_{-j}')=P(s_j'|s_j,a_j,a_{-j})$, the inner conditional expectation in \eqref{preq:expectationRewriteSelf} is independent of agent $i$'s policy $\theta_i$ pointwise, given $(s_{-j},a_{-j})$. That is, when $a_{i,t}$ changes to $a_{i,t}'$, the trajectory of any other agent $j$ remains unaffected. Therefore, Eq. \eqref{eq:rewrittenMPGDefinition} holds, completing the proof.
\end{proof}
\begin{remark}
   We note that the condition $P(s_j'|s_j,a_j,a_{-j}')=P(s_j'|s_j,a_j,a_{-j})$ means that the state transition of agent $i$ is independent of other agents' actions, and that the vehicles state transition model \eqref{eq:vehicleDynamics} satisfies this condition. We also comment that the reward design \eqref{eq:selfReward} can be used to model self-centered driving objectives, such as tracking a desired speed, maximizing ride comfort, or optimizing self-energy efficiency. These types of driving objectives only depend on the actions/policies of the ego vehicle and thus satisfy \eqref{eq:selfReward}. 
\end{remark}

Theorem \ref{thm:jointPotentialFunction} states that if agents' interactions are pairwise symmetrical and if the state transition condition holds, then the game is an MPG. 

\begin{thm}\label{thm:jointPotentialFunction}
    Consider an MG where agent $i$'s reward function is of the form,
    \begin{equation}\label{eq:jointReward}
        r_i(s_t,a_t)=\sum_{j\in\mathcal{N},j\neq i}r_{ij}(s_{i,t},s_{j,t},a_{i,t},a_{j,t}),
    \end{equation}
    where $r_{ij}(s_{i,t},s_{j,t},a_{i,t},a_{j,t})=r_{ji}(s_{j,t},s_{i,t},a_{j,t},a_{i,t})$, $\forall i,$ $j\in \mathcal{N}, i\neq j$. Suppose $P(s_i'|s_i,a_i,a_{-i}')=P(s_i'|s_i,a_i,a_{-i})$, $\forall a_{-i},a_{-i}'\in\mathcal{A}_{-i},\forall a_i\in\mathcal{A}_i,\forall s_i\in\mathcal{S}_i$ and $\forall i\in\mathcal{N}$. Then the formulated game is an MPG with a potential function  \begin{equation}\label{eq:jointPotenitalFunction}
        \phi^{\mathrm{joint}}(s_t,a_t)=\sum_{i\in\mathcal{N}}\sum_{j\in\mathcal{N},j<i}r_{ij}(s_{i,t},s_{j,t},a_{i,t},a_{j,t}).
    \end{equation}
\end{thm}

\begin{proof}
    From \eqref{eq:jointReward}, the total reward of agent $i$ is given by:
    \begin{equation}\label{preq:jointTotalReward}
    \begin{split}
          &J_i(\theta)\\
          &=\mathbb{E}_{s_0\sim\rho}\Biggl[\sum_{t=0}^{\infty}\gamma^{t}\sum_{j\in\mathcal{N},j\neq i}r_{ij}(s_{i,t},s_{j,t},a_{i,t},a_{j,t})\Bigg|\pi_{\theta},s_0\Biggr]\\
          &=\frac{1}{1-\gamma}\mathbb{E}_{(s,a)\sim \mu_{\theta}(s,a)}\Big[\sum_{j\in\mathcal{N},j\neq i}r_{ij}(s_{i},s_{j},a_{i},a_{j})\Big].
    \end{split}
    \end{equation}
    
A unilateral deviation in agent $i$'s policy leads to the following difference in total reward,
\begin{small}
\begin{equation}\label{preq:leftJointRewardFunction}
    \begin{split}
        &J_i(\theta_i',\theta_{-i}) - J_i(\theta_i,\theta_{-i})\\
        &= \mathbb{E}_{s_0\sim\rho}\Biggl\{ \sum_{t=0}^{\infty}\gamma^{t} \Biggl[
        \sum_{j\in\mathcal{N},\, j\neq i} \Bigl(
        r_{ij}(s_{i,t}',s_{j,t},a_{i,t}',a_{j,t})\\
        &\quad - r_{ij}(s_{i,t},s_{j,t},a_{i,t},a_{j,t})
        \Bigr)\Bigg|\pi_{(\theta_i',\theta_{-i})},\pi_{(\theta_i,\theta_{-i})},s_0\Biggr]\Biggr\}\\
        &=\frac{1}{1-\gamma}\Bigg(\mathbb{E}_{(s,a)\sim \mu_{\theta_i',\theta_{-i}}(s,a)}\Big[\sum_{j\in\mathcal{N},j\neq i}r_{ij}(s_{i},s_{j},a_{i},a_{j})\Big]\\
        &\quad -\mathbb{E}_{(s,a)\sim d_{\theta_i
        ,\theta_{-i}}(s,a)}\Big[\sum_{j\in\mathcal{N},j\neq i}r_{ij}(s_{i},s_{j},a_{i},a_{j})\Bigg).
    \end{split}
    \end{equation}
\end{small}
    
    From \eqref{eq:selfPotenitalFunction}, the total potential function is given by:
    \begin{small}
    \begin{equation}\label{preq:jointPotentialFunction}
    \begin{split}
    &\Phi(\theta)\\
    &=\mathbb{E}_{s_0\sim\rho}\Biggl[\sum_{t=0}^{\infty}\gamma^{t}\sum_{i\in\mathcal{N}}\sum_{j\in\mathcal{N},j<i}r_{ij}(s_{i,t},s_{j,t},a_{i,t},a_{j,t})\Bigg|\pi_{\theta},s_0\Biggr]\\
    &=\frac{1}{1-\gamma}\,
\mathbb{E}_{(s,a)\sim \mu_{\theta}(s,a)}\Bigg[\sum_{i\in\mathcal{N}}\sum_{j\in\mathcal{N},j<i}r_{ij}(s_{i},s_{j},a_{i},a_{j})\Bigg].
    \end{split}
    \end{equation}
    \end{small}
    
    The same deviation shall yield the following difference in the total potential function:
    \begin{small}\begin{equation}\label{preq:rightJointPotentialFunction}
    \begin{split}
        &\Phi(\theta_i',\theta_{-i})-\Phi(\theta_i,\theta_{-i})\\
        &=\frac{1}{1-\gamma}\Bigg(\mathbb{E}_{(s,a)\sim \mu_{\theta_i',\theta_{-i}}(s,a)}\Bigg[\sum_{i\in\mathcal{N}}\sum_{j\in\mathcal{N},j<i}r_{ij}(s_{i},s_{j},a_{i},a_{j})\Bigg]\\
        &\quad -\mathbb{E}_{(s,a)\sim \mu_{\theta_i,\theta_{-i}}(s,a)}\Bigg[\sum_{i\in\mathcal{N}}\sum_{j\in\mathcal{N},j<i}r_{ij}(s_{i},s_{j},a_{i},a_{j})\Bigg]\Bigg).
    \end{split}
    \end{equation}
    \end{small}

    For \eqref{preq:leftJointRewardFunction} and \eqref{preq:rightJointPotentialFunction} to align, we recall the symmetry of the joint reward function, i.e., $r_{ij}(s_{i,t},s_{j,t},a_{i,t},a_{j,t})=r_{ji}(s_{j,t},s_{i,t},a_{j,t},a_{i,t})$, and prove the following:
    \begin{equation}\label{preq:expectationLevelJoint}
        \begin{split}
            &\mathbb{E}_{(s,a)\sim \mu_{\theta_i',\theta_{-i}}(s,a)}\Bigg[\sum_{j\in-i}\sum_{k\in-i,k<j}r_{jk}(s_{j},s_{k},a_{j},a_{k})\Bigg]\\
            &=\mathbb{E}_{(s,a)\sim \mu_{\theta_i,\theta_{-i}}(s,a)}\Bigg[\sum_{j\in-i}\sum_{k\in-i,k<j}r_{jk}(s_{j},s_{k},a_{j},a_{k})\Bigg].
        \end{split}
    \end{equation}

    We define:
    \begin{equation}
    \begin{split}
        s_{-(j,k)} := s_\ell,\quad \ell\in \mathcal N\setminus\{j,k\},\\
a_{-(j,k)} := a_\ell,\quad \ell\in \mathcal N\setminus\{j,k\}.
    \end{split}
    \end{equation}

    Under the law of total expectation, we rewrite \eqref{preq:expectationLevelJoint} as:
    \begin{small}\begin{equation}\label{preq:expectationRewriteJoint}
        \begin{split}
            &\mathbb{E}_{(s,a)\sim \mu_{\theta_i',\theta_{-i}}(s,a)}\\
            &\Bigg[\mathbb{E}\Bigg[\sum_{j\in-i}\sum_{k\in-i,k<j}r_{jk}(s_{j},s_{k},a_{j},a_{k})\Bigg|s_{-(j,k)},a_{-(j,k)}\Bigg]\Bigg],\\
            &=\mathbb{E}_{(s,a)\sim \mu_{\theta_i,\theta_{-i}}(s,a)}\\
            &\Bigg[\mathbb{E}\Bigg[\sum_{j\in-i}\sum_{k\in-i,k<j}r_{jk}(s_{j},s_{k},a_{j},a_{k})\Bigg|s_{-(j,k)},a_{-(j,k)}\Bigg]\Bigg].
        \end{split}
    \end{equation}\end{small}

    Similarly, under the assumption of $P(s_j'|s_j,a_j,a_{-j}')=P(s_j'|s_j,a_j,a_{-j})$, the inner expectation in \eqref{preq:expectationRewriteJoint} is independent of agent $i$'s policy $\theta_i$, and thus Eq. \eqref{preq:expectationRewriteJoint} holds.
\end{proof}

\begin{remark}
We comment that the reward design \eqref{eq:jointReward} can be used to model symmetric agents' interactions, e.g., a symmetric penalty to both vehicles $i$ and $j$ if a collision between them happens. 
\end{remark}

Theorem~\ref{thm:selfAndJointPotentialFunction} shows that a linear combination of the above two types of reward also lead to an MPG. 

\begin{thm}\label{thm:selfAndJointPotentialFunction}
    Consider an MG where agent $i$'s reward function is of the form,
    \begin{equation}\label{eq:mixReward}
    \begin{split}
        &r_i(s_t,a_t)\\
        &=\alpha_i r_i^{\mathrm{self}}(s_{i,t},a_{i,t})+\sum_{j\in\mathcal{N},j\neq i}\beta_{ij}r_{ij}(s_{i,t},a_{i,t},s_{j,t},a_{j,t}),
    \end{split}
    \end{equation}
    \noindent where $r_i^{\mathrm{self}}(s_{i,t},a_{i,t})$ and $\sum_{j\in\mathcal{N},j\neq i}r_{ij}(s_{i,t},a_{i,t},s_{j,t},a_{j,t})$ follow  from \eqref{eq:selfReward} and \eqref{eq:jointReward}, respectively, and $\alpha_i\in\mathbb{R}$ and $\beta_{ij}\in\mathbb{R}$. Suppose $P(s_i'|s_i,a_i,a_{-i}')=P(s_i'|s_i,a_i,a_{-i})$, $\forall a_{-i},a_{-i}'\in\mathcal{A}_{-i},\forall a_i\in\mathcal{A}_i,\forall s_i\in\mathcal{S}_i$ and $\forall i\in\mathcal{N}$. Then the MG is an MPG with a potential function
    \begin{small}\begin{equation}\label{eq:selfAndJointPotentialFunction}
    \begin{split}
        &\phi(s_t,a_t)\\
        &=\alpha_i \sum_{i\in\mathcal{N}}r_i^{\mathrm{self}}(s_{i,t},a_{i,t})+\sum_{i\in\mathcal{N}}\sum_{j\in\mathcal{N},j< i}\beta_{ij}r_{ij}(s_{i,t},s_{j,t},a_{i,t},a_{j,t})).
    \end{split}
    \end{equation}\end{small}
\end{thm}
\begin{proof}
    The result follows directly from Theorem \ref{thm:selfPotentialFunction} and Theorem \ref{thm:jointPotentialFunction}.
\end{proof}
\begin{remark}
Theorem \ref{thm:selfAndJointPotentialFunction} suggests that if a vehicle's driving objective can be modeled as the combination of self-centered terms, e.g., tracking a desired speed, and interactive terms, e.g., pairwise collision penalties, then the resulting game is still an MPG.
\end{remark}
\section{Multi-agent reinforcement learning design}
This section designs the MARL algorithm and policy neural networks (NNs) to solve the MPG for AD applications.

\begin{algorithm}
\caption{Potential-function-based Gradient Ascent.}
\label{alg:PFGradientPlay}
 \textbf{Require:}  step size $\eta$, time horizon $T$, number of epochs $E$
\begin{algorithmic}[1]
\State Find the total potential function according to \eqref{eq:selfAndJointPotentialFunction}.
\State Sample $s_0\sim\rho$.
\State Initialize shared policy parameter $\theta_s^{(0)}$.
\For{$k=0,1,\cdots,E-1$}
    \State \parbox[t]{0.85\linewidth}{Implement policy $\theta_s^{(k)}$ and collect trajectories.}
    \State \parbox[t]{0.95\linewidth}{Perform projected gradient ascent as in \eqref{eq:gradientPlayTotal}.}
\EndFor
\end{algorithmic}
\end{algorithm}
To solve the constructed MPGs, we consider the potential-function-based gradient ascent (i.e., \eqref{eq:gradientPlayTotal}) along with a parameter-sharing network structure. The learning algorithm is detailed in Algorithm \ref{alg:PFGradientPlay}. For the parameter-sharing network,
we let each vehicle $i$ adopt a deterministic policy $\pi_i$ parameterized by a NN. Recall from Section~\ref{subsec:whatAndWhy} that gradient play \eqref{eq:gradientPlayTotal}
maintains a separate parameterization $\theta_i$ for each agent $i \in \mathcal{N}$, yielding a number of parameters that grows linearly with $N$.
In our setting, all vehicles are modeled with identical kinematics, as in \eqref{eq:vehicleDynamics}. This homogeneity enables a parameter-sharing scheme \cite{parameterSharing, scalingWithMARL, Surprising}, in which a single policy network with shared parameters $\theta_s \in \mathcal{X}_s$ is used across all agents.
The per-agent update \eqref{eq:gradientPlayTotal} then becomes:
\begin{equation}\label{eq:gradientPlayTotalShare}
    \theta_s^{(k+1)}
    = \mathrm{Proj}_{\mathcal{X}_s}\left(
        \theta_s^{(k)} + \eta\,\nabla_{\theta_s}\Phi\left(\theta_s^{(k)}\right)
    \right),
    \quad \eta > 0.
\end{equation}

Such a parameter-sharing NN is employed to improve training scalability, especially for large-scale systems.

For the NN architecture, we let each vehicle $i$ receive a local observation $o_i = O_i(s)$, where $O_i : \mathcal{S} \rightarrow \mathcal{O}_i$ is a mapping from the global state space to each vehicle $i$'s local observation space.
As illustrated in Figure~\ref{fig:NNArchitecture}, the policy network receives the local observation $o_i$, and outputs an action $a_i$. For clarity, the figure illustrates a single action output. However, the output can be extended to produce multiple actions if required. The architecture consists of three layers:
\begin{enumerate}
    \item \textbf{Input layer}: receives the local observation as the input for each agent.
    \item \textbf{Hidden layers}: two fully connected layers with Leaky Rectified Linear Unit (Leaky ReLU) activations, to introduce nonlinearity.
    \item \textbf{Output layer}: a fully connected layer with a Hyperbolic Tangent (Tanh) 
          activation, which constrains the raw output to $[-1, 1]$; this can be subsequently 
          scaled to the physical action range.
\end{enumerate}

\begin{figure}
\captionsetup[subfloat]{font=footnotesize}
    \centering
    \includegraphics[width=0.9\linewidth]{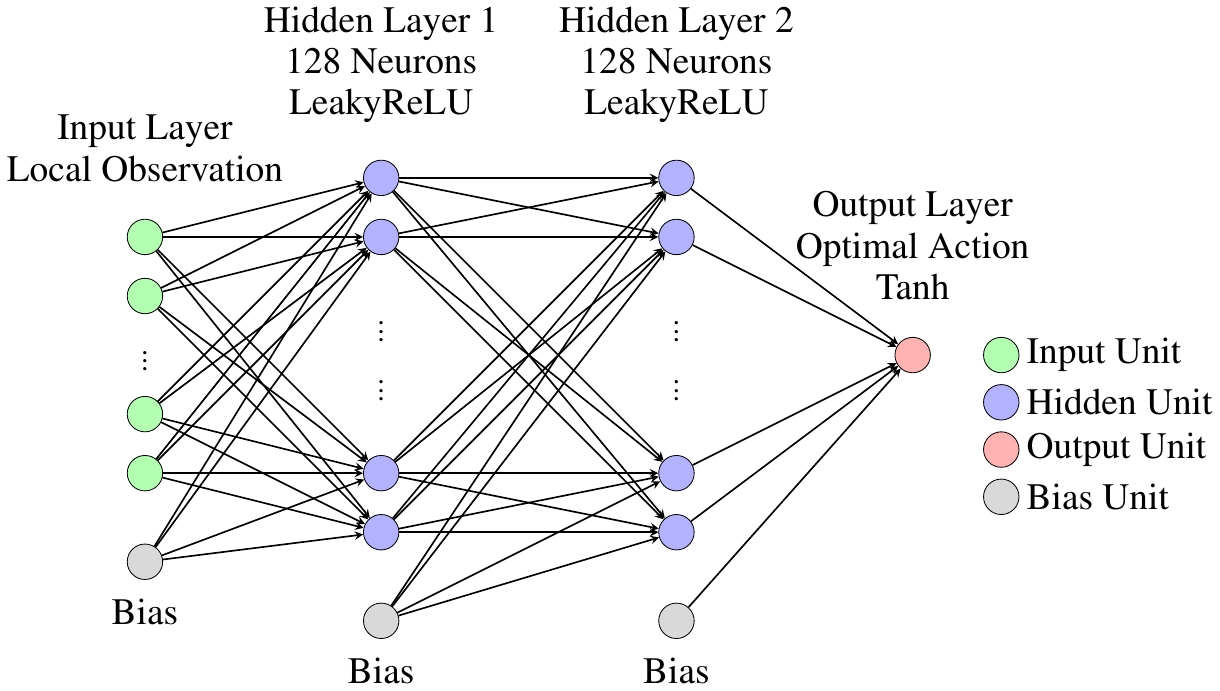}
    \caption{Parameter-sharing policy network architecture.}
    \label{fig:NNArchitecture}
\end{figure}

\section{Numerical results in highway forced merge scenarios}
This section considers highway forced merge scenarios as illustrative examples to validate the MPG-based MARL algorithm in autonomous driving applications.
\subsection{Simulation setup}
\vspace{0.25em}
\textbf{(i) Environment and agent configuration}\par
\vspace{0.25em}
Consider the highway forced merge scenario shown in Figure~\ref{fig:HMMPG}, where the on-ramp vehicle is the ego vehicle, and the others are surrounding vehicles traveling in the target lane. Among the target-lane vehicles, four leading vehicles and four following vehicles are selected as game players (therefore $N=9$), where ``leading'' and ``following'' vehicles are defined according to their longitudinal positions relative to the ego vehicle.

\begin{figure}
\captionsetup[subfloat]{font=footnotesize}
    \centering
    \includegraphics[width=0.9\linewidth]{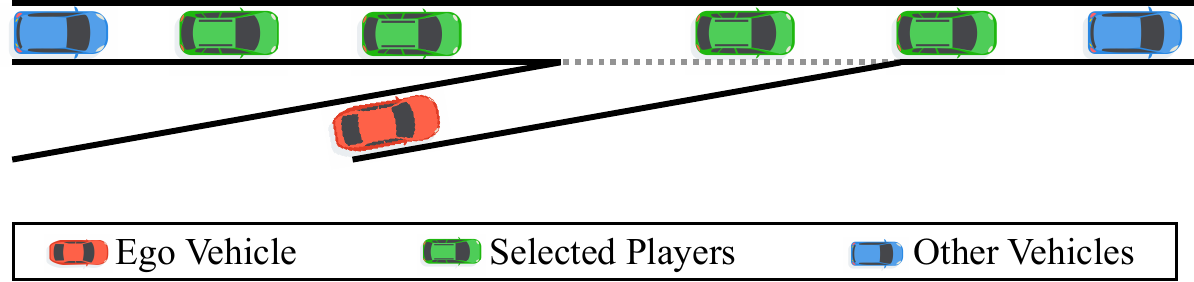}
    \caption{Single-lane highway forced merge scenario.}
    \label{fig:HMMPG}
\end{figure}

For each vehicle, the action is designed to be the longitudinal acceleration, and the system state (positions and velocities) evolves according to the vehicle dynamics \eqref{eq:vehicleDynamics}. We set the speed of each vehicle be within the range $0~m/s\leq v_{i,t}\leq30~m/s$.

For the parameter-sharing policy NN, due to the asymmetries of the on-ramp and target-lane vehicles, we introduce two additional input signals for the forced-merge driving:
\begin{itemize}
\item \textbf{Lane flag}: a binary indicator of whether vehicle $i$ is on the ramp or the target lane.
\item \textbf{Agent index}: a one-hot vector identifying vehicle $i$ within the player set.
\end{itemize}
The complete local observation is listed in Table~\ref{tab:observation}.

\begin{table}[htbp]
\centering
\caption{Local observation for each vehicle $i$.}
\label{tab:observation}
\begin{tabular}{c|c}
\hline
Feature & Definition \\ \hline
Position & \makecell{Longitudinal distance to the \\ merging conflict point} \\ \hline
Speed             & Speed                                       \\ \hline
Leader distance    & Longitudinal distance to immediate leader      \\ \hline
Leader relative speed    & Relative speed w.r.t.\ immediate leader              \\ \hline
Leader flag          & Indicator for leader presence                  \\ \hline
Follower distance  & Longitudinal distance to immediate follower         \\ \hline
Follower relative speed  & Relative speed w.r.t.\ immediate follower            \\ \hline
Follower flag        & Indicator for follower presence                \\ \hline
Lane flag            & Indicator for on-ramp or target-lane           \\ \hline
Agent index            & One-hot vector for vehicle identity       \\ \hline
\end{tabular}
\end{table}
\begin{figure}
\captionsetup[subfloat]{font=footnotesize}
    \centering
    \includegraphics[width=0.71\linewidth]{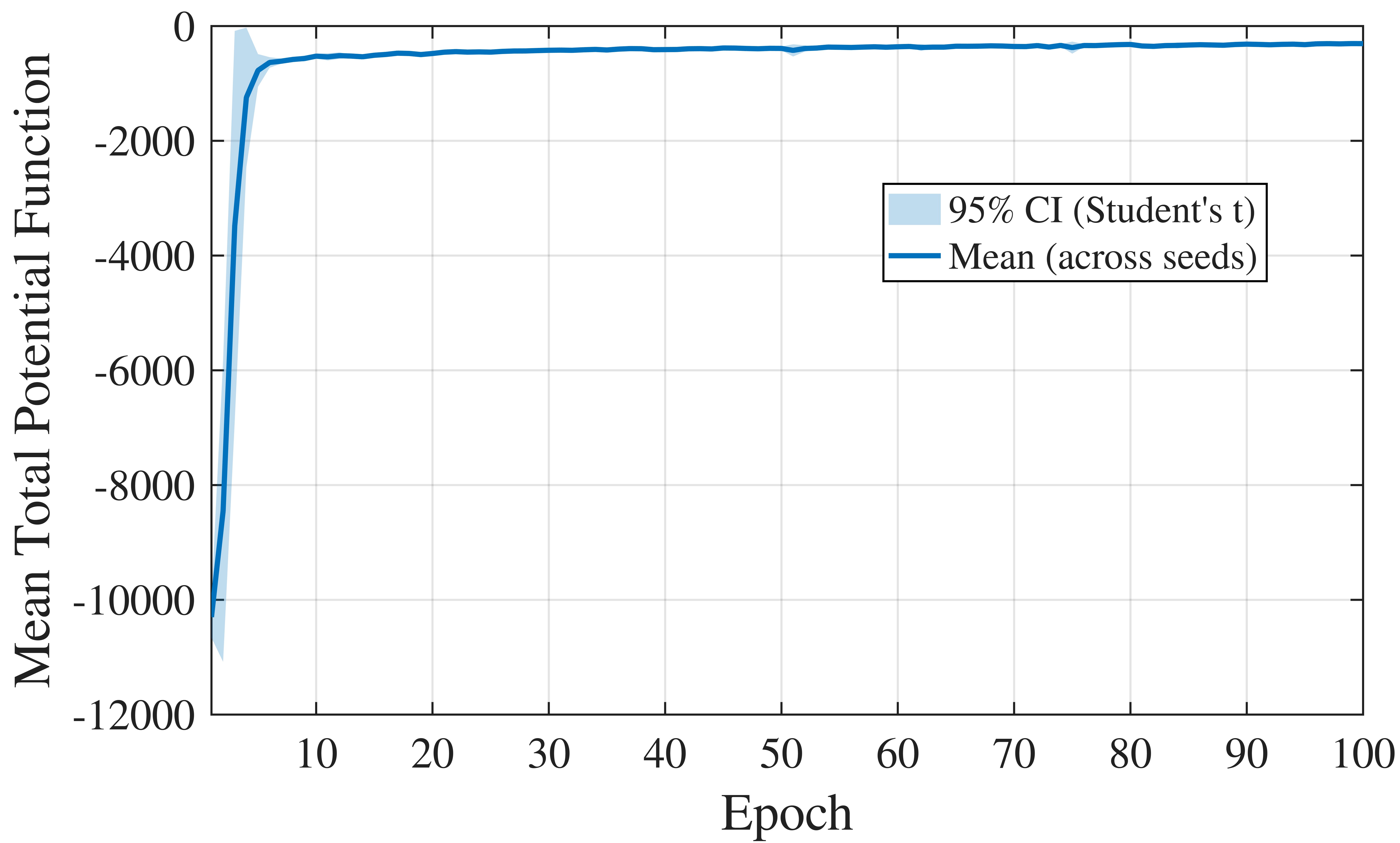}
    \caption{Convergence of the mean potential function during training.}
    \label{fig:convergence}
\end{figure}
\begin{figure}
\captionsetup[subfloat]{font=footnotesize}
    \centering
    \includegraphics[width=0.73\linewidth]{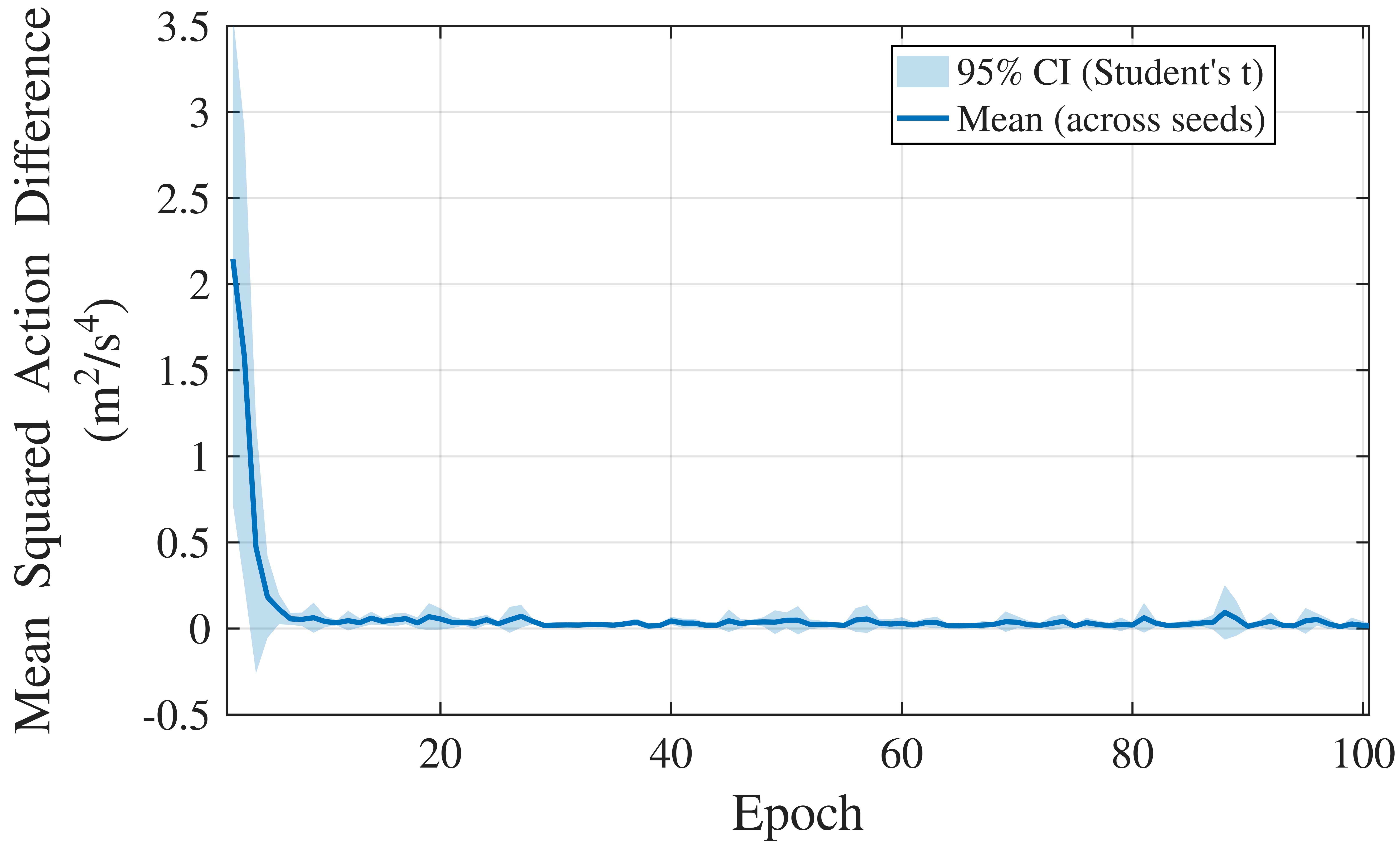}
    \caption{Convergence of the mean squared action difference during training.}
    \label{fig:actionStability}
\end{figure}

\begin{figure*}[!t]
\captionsetup[subfloat]{font=footnotesize}
\centering
\subfloat[]{
  \centering
  \includegraphics[width=0.95\textwidth,height=0.20\textheight,keepaspectratio]{NE11.png}
  \label{fig:NE11}
}
\par
\subfloat[]{
  \centering
  \includegraphics[width=0.95\textwidth,height=0.20\textheight,keepaspectratio]{NE12.png}
  \label{fig:NE12}
}
\par
\subfloat[]{
  \centering
  \includegraphics[width=0.95\textwidth,height=0.20\textheight,keepaspectratio]{NE13.png}
  \label{fig:NE13}
}
\caption{Illustration of the merging behavior in Scenario 1 of the simulated data. (a) Initial state. (b) Ego vehicle before merging. (c) Ego vehicle after merging.}
\label{fig:NE1}
\end{figure*}

\vspace{0.25em}
\textbf{(ii) Action space and feasibility sets}\par
\vspace{0.25em}
The action space for each vehicle \(i\) is \(\mathcal{A}_i=[-g,g]\), where \(g=9.81~\mathrm{m/s^2}\). For a target-lane vehicle, an action \(a_i\in\mathcal{A}_i\) is considered feasible if with this action, the time-to-collision (TTC) between vehicle \(i\) and its immediate leader and between vehicle \(i\) and its immediate follower are both larger than the safety threshold \(\tau_s=3~\mathrm{s}\), where TTC is defined as
\[
\mathrm{TTC}=\frac{d}{\Delta v},
\]
where \(d\) denotes the longitudinal distance and \(\Delta v\) denotes the longitudinal speed difference. If \(\Delta v\le 0\), the TTC is taken as \(+\infty\). Accordingly, the feasible action set \(\mathcal{A}_i^{\mathrm{feas}}\) is defined as the subset of \(\mathcal{A}_i\) satisfying these safety constraints.

\vspace{0.25em}
\textbf{(iii) Reward function design}\par
\vspace{0.25em}
The reward of each vehicle comprises four parts: desired speed tracking, ride comfort, and collision avoidance.
Specifically, 
\begin{align}
    J_i(\theta_s) &= \omega_{i,11}J_i^{\mathrm{self},1}(\theta_s)+\omega_{i,12}J_i^{\mathrm{self},2}(\theta_s)\\
    &\quad+\omega_{i,21}J_i^{\mathrm{joint},1}(\theta_s)+\omega_{i,22}J_i^{\mathrm{joint},2}(\theta_s),
\end{align}
where $\omega_{i,11}$, $\omega_{i,12}$, $\omega_{i,21}$, and $\omega_{i,22}$ are constant real numbers to balance the four terms.

The first term, i.e., $J_i^{\mathrm{self},1}(\theta_s)$, is to motivate each vehicle to track a desired speed, and it takes the form \eqref{preq:selfTotalReward}, where $r_i^{\mathrm{self}}(s_{i,t},a_{i,t})$ is given by
\begin{equation}
    r_i^{\mathrm{self}}(s_{i,t},a_{i,t})=-(v_{i,t}-v_{i,d})^2,
\end{equation}
where $v_{i,d}=15~m/s$ is the desired speed for each vehicle. 

The second term $J_i^{\mathrm{self},2}(\theta_s)$ takes the form \eqref{preq:selfTotalReward} and is designed to encourage ride comfort \cite{PCPGMLiu}:
\begin{equation}
    r_i^{\mathrm{self}}(s_{i,t},a_{i,t})=-(u_{i,t})^2.
\end{equation}

The third and fourth terms are formulated to prevent collisions, and they both take the form \eqref{preq:jointTotalReward}. The third term rewards larger $\mathrm{TTC}$ for same-lane interactions and $r_{ij}(s_{i,t},s_{j,t},a_{i,t},a_{j,t})$ is given by
\begin{equation}
\begin{aligned}
&r_{ij}(s_{i,t},s_{j,t},a_{i,t},a_{j,t}) = \\[0.2em]
&\begin{cases}
  -\dfrac{1}{\dfrac{|x_{i,t}-x_{j,t}|}{v_c}+\epsilon},
     &|v_{i,t}-v_{j,t}| \le v_c\ \text{m/s}, \\[20pt]
  -\dfrac{1}{\mathrm{TTC}_{ij}+\epsilon},
     &|v_{i,t}-v_{j,t}| > v_c\ \text{m/s},
\end{cases}
\end{aligned}
\label{eq:piecewise}
\end{equation}
where $\mathrm{TTC}_{ij}=\dfrac{|x_{i,t}-x_{j,t}|}{|v_{i,t}-v_{j,t}|}$ and $v_c$ denotes a constant relative-speed threshold, and $\epsilon$ is a very small positive number to avoid the denominator being $0$. 

The fourth term rewards larger time gaps to the merging conflict point $x_c=180~\mathrm{m}$ for different-lane vehicles, where the choice for the merging conflict point is addressed in Section~\ref{subsec:dataset}. The reward $r_{ij}(s_{i,t},s_{j,t},a_{i,t},a_{j,t})$ is given by
\begin{equation}\label{simeq:jointReward2}
\begin{split}
    &r_{ij}(s_{i,t},s_{j,t},a_{i,t},a_{j,t})\\
          &=-\frac{1}{\sqrt{T_i(t)T_j(t)}(T_i(t)-T_j(t))^2+\epsilon},
\end{split}
\end{equation}
where $T_i(t)=\dfrac{|x_{i,t}-x_c|}{v_{i,t}+\epsilon}$ and $T_j(t)=\dfrac{|x_{j,t}-x_c|}{v_{j,t}+\epsilon}$ denote the time for vehicles $i$ and $j$ to reach the conflict point. The term $\sqrt{T_i(t)T_j(t)}$ is added to weight the TTC difference with a ``risk level": given $(T_i(t)-T_j(t))$, the smaller the $\sqrt{T_i(t)T_j(t)}$ is, the more dangerous the situation is.  

According to Theorem \ref{thm:selfAndJointPotentialFunction}, the above reward design leads to am MPG. 

\vspace{0.25em}
\textbf{(iv) Training setup and results}\par
\vspace{0.25em}
We train the shared policy network by maximizing the total potential function over horizon $T$,
\begin{equation}\label{simeq:totalPotentialFunction}
\Phi(\theta_s)=\mathbb{E}_{s_0\sim\rho}\left[\sum_{t=0}^{T}\gamma^t\phi(s_t,a_t)\Biggm|\pi_{\theta_s},s_0\right],
\end{equation}
where we select $\gamma=0.99$ and $T=30$ s. 

The initial states are generated using stratified sampling to ensure broad coverage of the state space. Vehicle longitudinal positions and velocities are sampled within predefined ranges while requiring a minimum longitudinal headway of 7~m and an initial TTC of at least 4~s for all successive leader and follower.

We train the trajectories induced by the initial states with 5 random seeds. Episodes will terminate when any pair of vehicles collide, i.e., when their bounding boxes overlap. The mean total potential function and the mean squared action difference across seeds during training are included in Figures \ref{fig:convergence} and \ref{fig:actionStability}, demonstrating the convergence of the policy NN during the training process.
\begin{figure*}[!t]
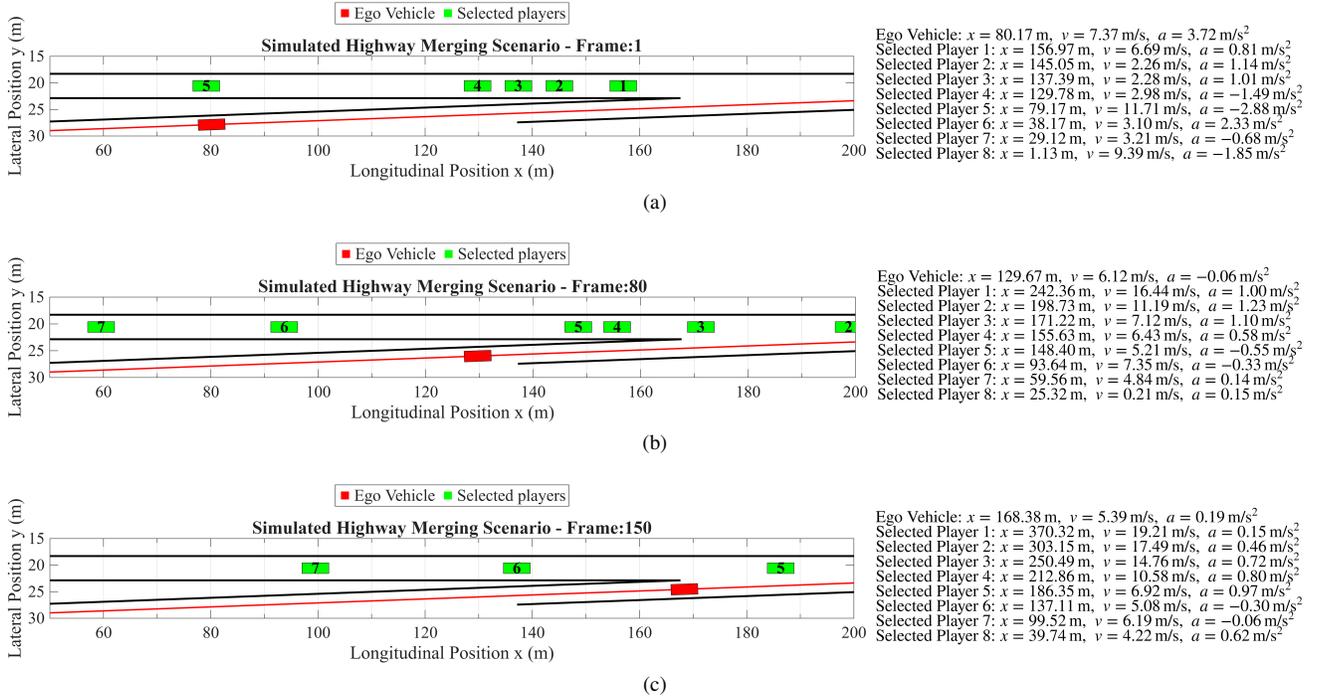

\captionsetup[subfloat]{font=footnotesize}
\centering
\subfloat[]{
  \centering
  \includegraphics[width=0.95\textwidth,height=0.20\textheight,keepaspectratio]{NE21.png}
  \label{fig:NE21}
}
\par
\subfloat[]{
  \centering
  \includegraphics[width=0.95\textwidth,height=0.20\textheight,keepaspectratio]{NE22.png}
  \label{fig:NE22}
}
\par
\subfloat[]{
  \centering
  \includegraphics[width=0.95\textwidth,height=0.20\textheight,keepaspectratio]{NE23.png}
  \label{fig:NE23}
}
\caption{Illustration of the merging behavior in Scenario 2 of the simulated data. (a) Initial state. (b) Ego vehicle before merging. (c) Ego vehicle after merging.}
\label{fig:NE2}
\end{figure*}
\begin{figure*}[!t]
\captionsetup[subfloat]{font=footnotesize}
\centering
\subfloat[]{
  \centering
  \includegraphics[width=0.95\textwidth,height=0.20\textheight,keepaspectratio]{Dataset11.png}
}
\par
\subfloat[]{
  \centering
  \includegraphics[width=0.95\textwidth,height=0.20\textheight,keepaspectratio]{Dataset12.png}
}
\par
\subfloat[]{
  \centering
  \includegraphics[width=0.95\textwidth,height=0.20\textheight,keepaspectratio]{Dataset13.png}
}
\caption{Illustration of the merging behavior in Scenario 1 of the NGSIM dataset. (a) Ego vehicle preparing for merging. (b) Ego vehicle before merging. (c) Ego vehicle after merging.}
\label{fig:Dataset1}
\end{figure*}
\begin{figure*}[!t]
\captionsetup[subfloat]{font=footnotesize}
\centering
\subfloat[]{
  \centering
  \includegraphics[width=0.95\textwidth,height=0.20\textheight,keepaspectratio]{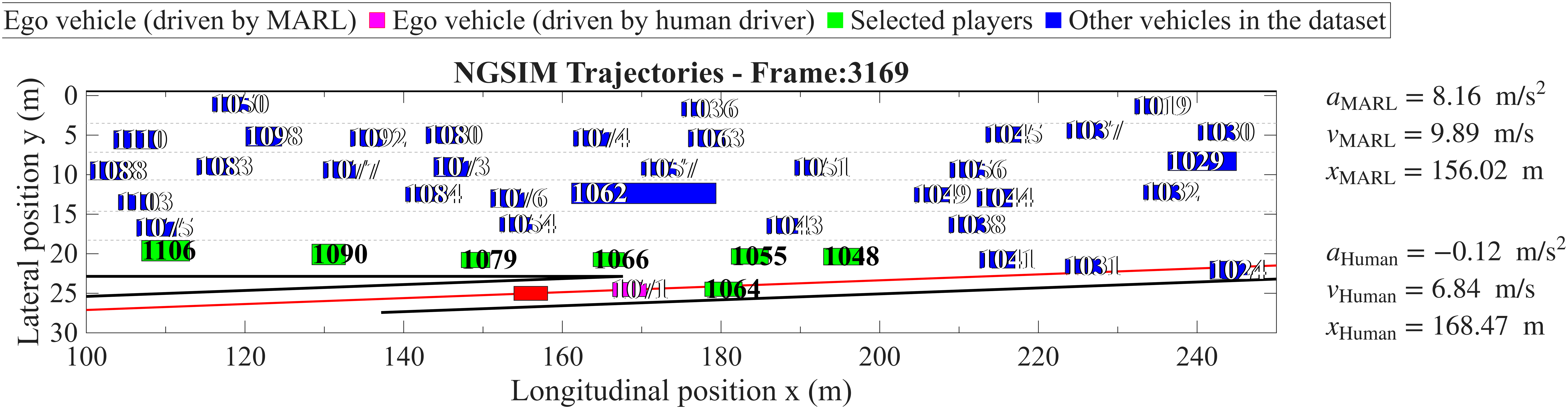}
}
\par
\subfloat[]{
  \centering
  \includegraphics[width=0.95\textwidth,height=0.20\textheight,keepaspectratio]{Dataset22.png}
}
\par
\subfloat[]{
  \centering
  \includegraphics[width=0.95\textwidth,height=0.20\textheight,keepaspectratio]{Dataset23.png}
}
\caption{Illustration of the merging behavior in Scenario 2 of the  NGSIM dataset. (a) Ego vehicle preparing for merging. (b) Ego vehicle before merging. (b) Ego vehicle after merging.}
\label{fig:Dataset2}
\end{figure*}
\begin{figure}[!t]
    \centering
    \includegraphics[width=0.73\linewidth]{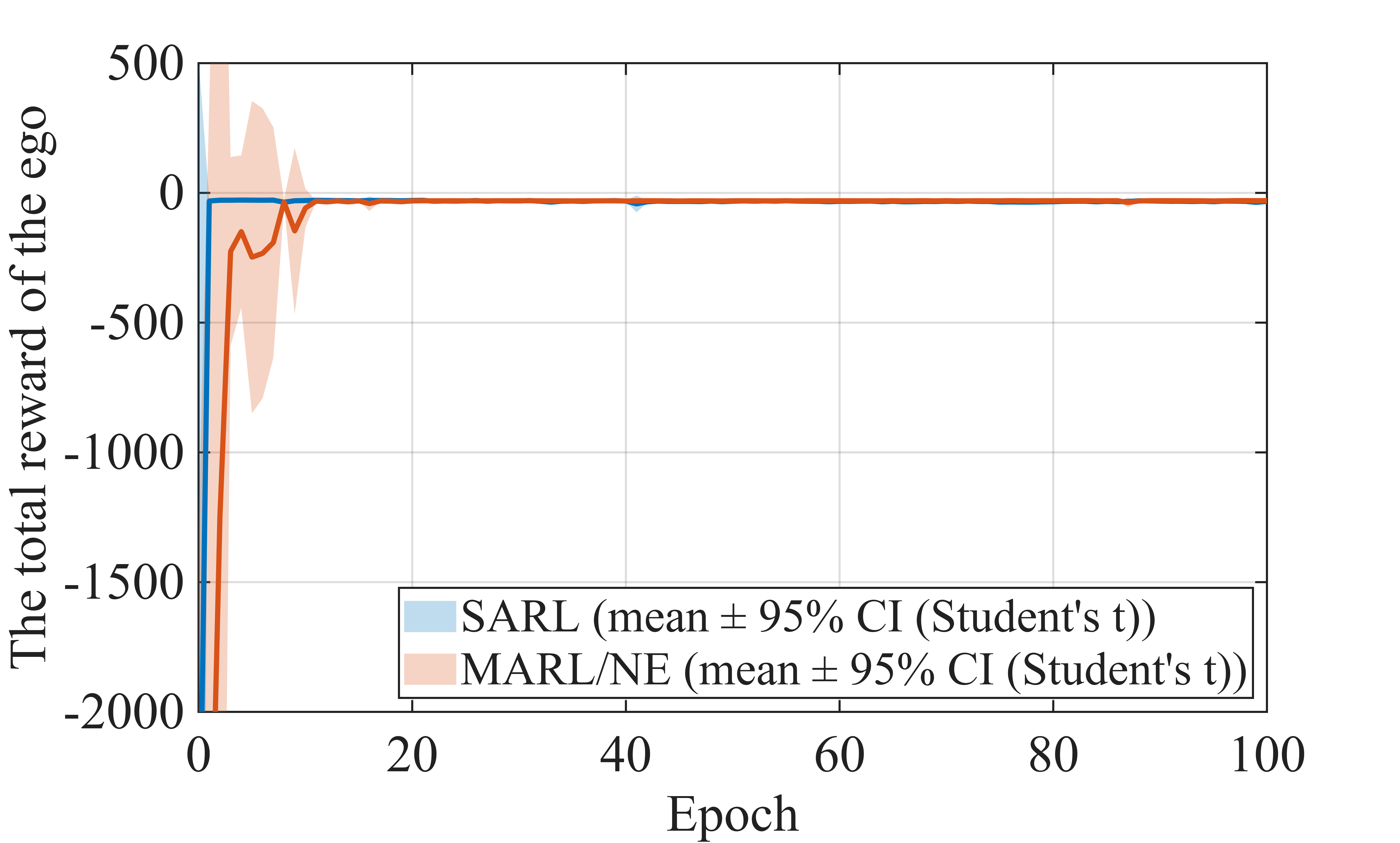}
    \caption{The total reward of the ego vehicle under training.}
    \label{fig:trainCmp}
\end{figure}
\begin{figure}[!t]
    \centering
    \includegraphics[width=0.73\linewidth]{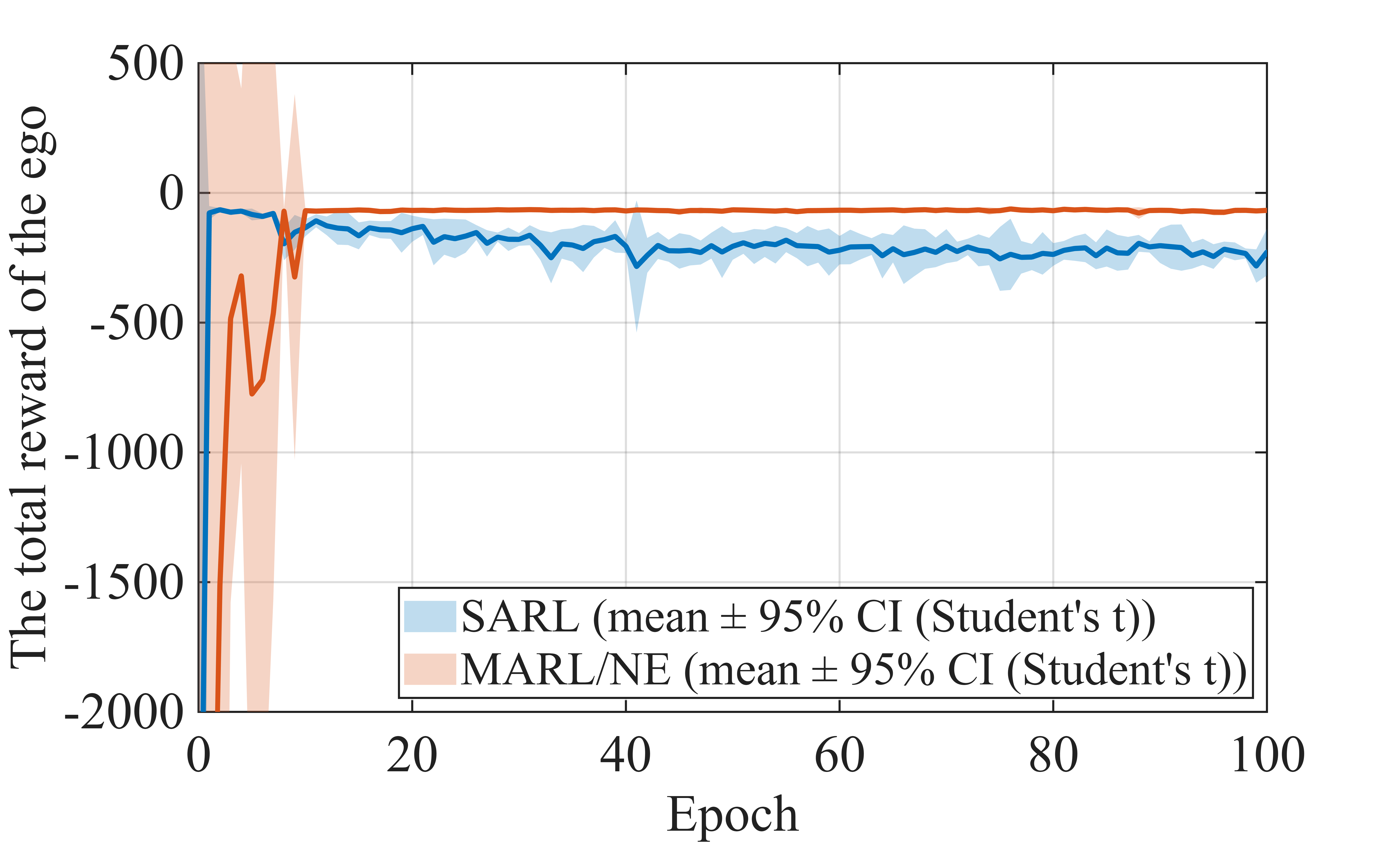}
    \caption{The total reward of the ego vehicle under testing.}
    \label{fig:testCmp}
\end{figure}
\subsection{Numerical results with simulated surrounding vehicles}\label{sec:simulated}
After deriving the policy NN, we then test it in  various merging scenarios to evaluate its driving performance. In this section, we consider simulated surrounding vehicles, i.e., their behaviors are generated from computer programs instead of human driving datasets. For this type of evaluations, we consider (i) \textbf{specific scenarios} to observe the ego vehicle performance in various situations, (ii) \textbf{statistical studies} to characterize the performance metrics, and (iii) \textbf{comparative studies} with the single-agent RL.

\vspace{0.25em}
\textbf{(i) Specific scenarios}\par
\vspace{0.25em}
\textit{Scenario 1:} In this scenario, the initial state is randomly chosen. The leading vehicles ``1” to ``4” accelerate, and the following vehicles ``5” to ``8” decelerate. Hence, the gap between vehicles ``4” and ``5” expands during driving. Three key moments are shown in Figure \ref{fig:NE1}. In this case, we observe that the trained policy NN enables the ego vehicle to safely merge into the target lane, as illustrated in Figure \ref{fig:NE13}. After the merging maneuver is completed, all vehicles tend to maintain their desired speeds.

\textit{Scenario 2:} In this scenario, compared with Scenario 1, we set an initially large gap between vehicles ``4" and ``5", as shown in Figure \ref{fig:NE21}. However, vehicle ``5" has a higher speed and the gap decreases rapidly. In this case, the ego vehicle first decelerates to wait for vehicle ``5" to pass and then accelerates to complete the merge. Three key moments are shown in Figure \ref{fig:NE2}. After the merging maneuver is completed, all vehicles tend to track their desired speeds. 

\vspace{0.25em}
\textbf{(ii) Statistical studies}\par
\vspace{0.25em}
We also conduct statistical studies to provide a comprehensive performance evaluation. 

To test the driving policy robustness, we consider three types of  surrounding vehicle policies: 1) the NE policies, 2) rule-based policies from the Intelligent Driver Model (IDM) \cite{IDM}, and 3) a constant-speed policy. The NE policies correspond to the idea case where all vehicles are perfectly rational. The IDM-based policies exhibit reactive, rule-driven behavior with limited anticipation of interactions. The constant-speed policies do not react to surrounding traffic conflicts, resulting in safety-agnostic or unresponsive driving (e.g., distracted or uncooperative driving).

We evaluate the derived policy NN over 500 randomly selected testing scenarios with various initial states. We then use six performance metrics to quantify the results, as summarized in Table \ref{tab:statisticalResultsMARL}. Here, the collision rate is the number of scenarios in which a collision involving the ego vehicle occurs over the total number of scenarios. The failure rate is the number of scenarios in which the ego vehicle fails to pass the conflict point due to dense traffic flow over the total number of scenarios. 
Note that these two quantities are not necessarily integers, as they are computed by averaging over multiple random seeds ($5$ random seeds in our simulation). The average minimum inter-vehicle distance is the average of the minimum distances in each of the 500 scenarios.

The statistical results shown in Table \ref{tab:statisticalResultsMARL} lead to the following observations:
\begin{itemize}
\item The trained policy NN can ensure the ego vehicle safety (i.e., $0$ collision rate) and successful merging (i.e., $0$ failure rate) when the surrounding vehicles take NN or IDM policies, demonstrating a certain level of robustness. 
    \item When the surrounding vehicles also take the NN policies, compared to the IDM-policy, the ego vehicle achieves larger safety margin (as reflected by the average minimum inter-vehicle distance) and higher travel efficiency (as reflected by the ego vehicle speed), while maintains similar energy efficiency (as reflected by the acceleration magnitude) and ride comfort (as reflected by the jerk magnitude), indicating a preference for the NE policy compared to the IDM policy. 
\end{itemize}
\begin{table*}[htbp]
\centering
\caption{Statistical Results: MARL with MPGs on Simulated Data}
\label{tab:statisticalResultsMARL}
\begin{tabular}{c|c|c|c} 
\hline
Selected players' policies & NE & IDM & Constant speed \\ \hline

Collision rate & 0/500 & 0/500 & 15.4/500\\ \hline

Failure rate & 0/500 & 0/500 & 0/500\\ \hline

Average minimum inter-vehicle distance (m) & 40.7305 & 13.6459 & 25.7577\\ \hline

Average ego vehicle speed (m/s) & 11.0292 & 5.7136 & 1.8372\\ \hline

Average acceleration magnitude$(\text{m}/\text{s}^{2})$ & 0.5017 & 0.5760 & 0.2167\\ \hline

Average jerk magnitude$(\text{m}/\text{s}^{3})$ & 0.2466 & 2.2602 & 0.0374\\ \hline
\end{tabular}
\end{table*}
\begin{table*}[htbp]
\centering
\caption{Statistical Results: Single-agent RL on Simulated Data}
\label{tab:statisticalResultsSARL}
\begin{tabular}{c|c|c|c} 
\hline
Surrounding vehicles' policies & NE & IDM & Constant speed \\ \hline

Collision rate & 0/500 & 0/500 & 15.6/500 \\ \hline

Failure rate & 0/500 & 0/500 & 0/500\\ \hline

Average minimum inter-vehicle distance (m) & 30.2177 & 14.8224 & 11.4778\\ \hline

Average ego vehicle speed (m/s) & 12.0469 & 9.7332 & 2.1027\\ \hline

Average acceleration magnitude$(\text{m}/\text{s}^{2})$ & 3.0121 & 2.9777 & 0.6665\\ \hline

Average jerk magnitude$(\text{m}/\text{s}^{3})$ & 11.2172 & 9.7573 & 1.1170\\ \hline
\end{tabular}
\end{table*}
\begin{table*}[!t]
\centering
\caption{Comparative Results on the Naturalistic Dataset}
\label{tab:datasetResults}
\begin{tabular}{c|c|c|c} 
\hline
Ego vehicle's policy & Human driver & Single-agent RL & MARL \\ \hline
Collision rate & 0/100 & 3.6/100 & 0/100 \\ \hline

Failure rate & 0/100 & 0/100 & 0/100 \\ \hline

Average minimum inter-vehicle distance (m) & 9.5395 & 10.5108 & 11.5422\\ \hline


Average ego vehicle speed (m/s) & 7.8863 & 6.8050 & 7.0615\\ \hline

Average acceleration magnitude$(\text{m}/\text{s}^{2})$ & 1.1258 & 1.7345 & 0.7092\\ \hline

Average jerk magnitude$(\text{m}/\text{s}^{3})$ & 2.6284 & 4.0728 & 0.7280\\ \hline

\end{tabular}
\end{table*}
~\\
\vspace{0.25em}
\indent\textbf{(iii) Comparative studies}\par
\vspace{0.25em}
Next, we consider comparative studies between single-agent RL and MARL. For single-agent RL training, we let the surrounding vehicles follow the IDM policy, and the ego vehicle learns its optimal policy using a gradient ascent algorithm (similar to \eqref{eq:gradientPlay} but with others' policies being fixed) using the same 5 random seeds and the same reward design as in MARL. 
The trained policies are then evaluated under the three surrounding-vehicle policies. The single-agent RL results are shown in Table \ref{tab:statisticalResultsSARL}. By comparing Tables \ref{tab:statisticalResultsMARL} and \ref{tab:statisticalResultsSARL}, it is observed that the policy trained from MARL exhibits better safety, energy efficiency, and ride comfort than the single-agent RL, as reflected by the smaller collision rates, smaller acceleration magnitude, and smaller jerk, respectively. 
We also present the ego vehicle's total rewards during training and testing (see Figures \ref{fig:trainCmp} and \ref{fig:testCmp}), which further confirms that the NE policy exhibits better robustness to drift in the initial state during testing, whereas single-agent RL tends to overfit the training data.
\subsection{Verification on the naturalistic driving dataset}\label{subsec:dataset}
This section evaluates the performance of the MPG-based MARL framework using naturalistic traffic data from the Federal Highway Administration’s Next Generation Simulation (NGSIM) I-80 dataset \cite{NGSIM_I80_2016}. In our simulation, both the roadway geometry and modeling parameters are selected to be consistent with empirical observations from the dataset. In particular, the prediction horizon is chosen to cover the longest merging duration observed in the data, and the conflict point is determined based on the road geometry. 

\vspace{0.25em}
\textbf{(i) Data processing}\par
\vspace{0.25em}
We process the raw data to reduce measurement noise and obtain smooth estimates of positions and velocities. For each vehicle, the corresponding time series is first sorted by frame index to ensure temporal consistency. The longitudinal and lateral positions are then smoothed using a Savitzky–Golay filter (implemented via MATLAB’s sgolayfilt function \cite{Savitzky-Golay, MATLAB2025a}). Vehicle velocities are recomputed by differentiating the smoothed positions. Lane indices are reassigned based on the vehicles’ lateral positions. Vehicle time series shorter than the filter window length are excluded from the smoothing procedure.

\vspace{0.25em}
\textbf{(ii) Specific scenarios}\par
\vspace{0.25em}
\textit{Scenario 1:} In this scenario, the red ego vehicle, controlled by the trained policy NN, starts together with the human-driven purple vehicle (whose movement is recorded in the dataset), with the same initial position and speed. The ego vehicle first decelerates to create a sufficient time gap before reaching the conflict point, then merges into the gap between vehicles ``1623” and ``1606,” and finally onto the target lane while maintaining a safe following distance. Three key moments are shown in Figure \ref{fig:Dataset1}. In contrast, the human-driven vehicle appears to attempt a merge into the small gap between vehicles ``1606” and ``1589.” However, it progresses aggressively, closing in on vehicle ``1589”, resulting in an abrupt deceleration later on.

\textit{Scenario 2:} In this scenario, the ego vehicle starts from the same longitudinal position as another human-driven vehicle. The ego vehicle identifies an appropriate merging opportunity between vehicles ``1079" and ``1066" and maintains this position throughout the maneuver. In contrast, the human-driven vehicle, based on the recorded data, initially attempts to merge between vehicles ``1066" and ``1064", but as vehicle ``1066" refuses to yield, it is ultimately forced back into the gap between ``1079" and ``1066". Three key moments are illustrated in Figure \ref{fig:Dataset2}.

\vspace{0.25em}
\textbf{(iii) Statistical and comparative studies}\par
\vspace{0.25em}
We further conduct comparative studies on the naturalistic dataset. Specifically, we let the ego vehicle be controlled by (i) human driver as recorded in the dataset, (ii) single-agent RL with the same setting as in the simulated data case (i.e., Section \ref{sec:simulated}), and (iii) MARL. The results are reported in Table \ref{tab:datasetResults}. From the tale, it is observed that:
\begin{itemize}
\item Both human-driven and MARL-controlled ego vehicles achieve successful merging in all 100 scenarios without collisions or failures, whereas the single-agent RL policy results in an non-zero collision rate.
\item The MARL achieves the largest safety margin, as reflected by the largest average minimum inter-vehicle distance.
\item Human drivers maintain the highest average speed, indicating the greatest travel efficiency. However, this is at the cost of safety: human drivers have the smallest inter-vehicle distance.
\item MARL achieves the best energy/fuel efficiency, as reflected by the smallest acceleration magnitude.
\item MARL achieves the best ride comfort, as reflected by the smallest jerk magnitude.
\end{itemize}
\section{Conclusion}
This paper studied MPG-based MARL and its use in autonomous driving. Sufficient conditions on the reward design and on the state transition probabilities are established to construct MPGs, which can ensure the NE attainability, including both the solution existence and the algorithm convergence, of the gradient play-based MARL. These conditions are shown to be able to accommodate general driving objectives, thus leading to MARL-based autonomous driving policy. Comprehensive evaluations in forced-merge driving scenarios were performed in both simulated and naturalistic traffic datasets. The results demonstrated that the learned driving policy can enable safe and efficient autonomous vehicles in dense traffic, outperforming the single-agent RL and the recorded human-driven vehicles in terms of safety, energy/fuel efficiency, and ride comfort. Future research will consider real-world autonomous vehicle testbeds to further verify the algorithm performance.
\bibliographystyle{IEEEtran}
\bibliography{Reference}
\end{document}